\documentclass[%
 reprint,
amsmath, amssymb,
aps,prx, twocolumn, superscriptaddress, showpacs
]{revtex4-1}

\def \bra#1{\mathinner{\langle{#1|}}}
\def \ket#1{\mathinner{|{#1}\rangle}}
\def \braket#1{\mathinner{\langle {#1} \rangle}}

\def \red #1 {\textcolor{red}{#1}}
\def \blue #1 {\textcolor{blue}{#1}}

\usepackage{hyperref}
\usepackage{enumitem} 
\hypersetup{
    colorlinks=true,
    linkcolor=blue,
    citecolor=blue,
    filecolor=blue,
    urlcolor=blue
}

\usepackage{graphicx}

\usepackage{dcolumn}
\usepackage{bm}
\usepackage{mathtools}
\usepackage{xcolor}
\usepackage{framed} 
\usepackage[all,cmtip]{xy}

\usepackage{algorithm}
\usepackage{algpseudocode}
\usepackage{float}

\newtheorem{theorem}{Theorem}

\newtheorem{proof}{Proof}
\newtheorem{corollary}{Corollary}

\newtheorem{observation}{Observation}

\usepackage{tikz}
\usetikzlibrary{automata, arrows}
\usepackage{circuitikz}
\usetikzlibrary{backgrounds,fit,decorations.pathreplacing,calc,backgrounds}  
\tikzstyle{operator} = [draw,fill=white,minimum size=1.5em] 
    \tikzstyle{operator2}=[draw,fill=white, text width=0.4cm, minimum height=2.5cm] 
    \tikzstyle{operator3}=[draw,fill=white, text width=0.8cm, minimum height=1cm] 
    \tikzstyle{phase} = [draw, fill=white,shape=circle,minimum size=12pt,inner sep=0pt]
    \tikzstyle{phase1} = [fill=black, shape=circle,minimum size=3.5pt,inner sep=0pt]
    \tikzstyle{phase2} = [fill=blue, shape=circle,minimum size=7pt,inner sep=1pt]
    \tikzstyle{phase11} = [fill=cyan, shape=circle,minimum size=7pt,inner sep=1pt]
    \tikzstyle{phase12} = [fill=lime, shape=circle,minimum size=7pt,inner sep=1pt]
    \tikzstyle{phase13} = [fill=pink, shape=circle,minimum size=7pt,inner sep=1pt]
    \tikzstyle{phase14} = [fill=lightgray, shape=circle,minimum size=7pt,inner sep=1pt]
    \tikzstyle{phase0} = [draw, fill=white, shape=circle,minimum size=3.5pt,inner sep=0pt]
    \tikzstyle{ellipsis} = [fill,shape=circle,minimum size=2pt,inner sep=0pt]
    \tikzset{meter/.append style={fill=white, draw, inner sep=3, font=\vphantom{A}, minimum width=10, path picture={\draw[black] ([shift={(.05,.2)}]path picture bounding box.south west) to[bend left=50] ([shift={(-.05,.2)}]path picture bounding box.south east);\draw[black,-latex] ([shift={(0,.15)}]path picture bounding box.south) -- ([shift={(.15,-.08)}]path picture bounding box.north);}}}
 \tikzset{cross/.style={path picture={\draw[thick,black](path picture bounding box.north) -- (path picture bounding box.south) (path picture bounding box.west) -- (path picture bounding box.east);
   }},
   crossx/.style={path picture={\draw[thick,black,inner sep=0pt]
   (path picture bounding box.south east) -- (path picture bounding box.north west) (path picture bounding box.south west) -- (path picture bounding box.north east);
   }},
   circlewc/.style={draw,circle,cross,minimum width=0.3 cm},
   }

\usepackage[small,bf]{caption}
\usepackage{tikz} 
\usetikzlibrary{backgrounds,fit,decorations.pathreplacing,calc,backgrounds}  

\begin{document}


\title{Floquet-informed Learning of Periodically Driven Hamiltonians}

\author{Keren Li}
\email{likr@szu.edu.cn}
\affiliation{College of Physics and Optoelectronic Engineering, Shenzhen University, Shenzhen 518060, China}
\affiliation{Quantum Science Center of Guangdong-Hong Kong-Macao Greater Bay Area (Guangdong), Shenzhen 518045, China}

\date{\today}

\begin{abstract}
Characterizing time-periodic Hamiltonians is pivotal for validating and controlling driven quantum platforms,  yet prevailing and unadjusted reconstruction methods demand dense time-domain sampling and heavy post-processing.
We introduce a scalable Floquet-informed learning algorithm that represents the Hamiltonian as a truncated Fourier series and recasts parameter estimation as a compact linear inverse problem in the Floquet band picture.
The algorithm is well-suited to problems satisfying mild smoothness/band-limiting. In this regime, its sample and runtime complexities scale polynomially with the Fourier cutoff, time resolution, and the number of unknown coefficients. For local Hamiltonian models, the coefficients grows polynomially with system size, yielding at most polylogarithmic dependence on the Hilbert-space dimension. 
Furthermore, numerical experiments on one- and two-dimensional Ising and Heisenberg lattices show fast convergence to time resolution and robustness to higher-order perturbations. An adaptive rule learns the cutoff of Fourier series, removing the need to set known truncation \textit{a priori}. 
These features enable practical certification and benchmarking of periodically driven platforms with rapidly decaying higher-order content, and extend naturally to near-periodic drives.
\end{abstract}

\maketitle

\paragraph*{Introduction.}
Periodically driven quantum systems supply an independent control axis, the drive frequency, amplitude, and waveform, that reshapes energy scales and symmetries. By exploiting stroboscopic averaging and micromotion, periodic driving engineers Hamiltonians with interactions and selection rules inaccessible in static settings, thereby advancing quantum simulation, control, and the study of non-equilibrium phases~\cite{Oka2019Floquet,eckstein2024large,Harper2020Topology}. Experimentally, the paradigm has enabled observations of anomalous Floquet insulators~\cite{rechtsman2013photonic}, non-equilibrium quantum many body dynamics~\cite{bordia2017periodically}, and discrete-time crystals~\cite{Else2016Floquet}. 
Furthermore, it has potentials to emulate Kitaev’s honeycomb model~\cite{kalinowski2023non}, realize ultrastrong light--matter coupling~\cite{Akbari2025Floquet}, and probe higher-order Weyl fermions~\cite{zhan2024perspective}. Together, these developments transform Floquet engineering from case-by-case demonstrations into a general approach for engineering Hamiltonians and tunable interactions, opening experimental access to phases and response functions previously out of reach.

Benchmarking such systems is essential. Without reliable reconstruction, small control imperfections accumulate, obscuring topological or time-crystalline signatures~\cite{rudner2020band} and degrading control protocols that assume accurate models~\cite{Viola1999Dynamical}. Central to this is the time-dependent Hamiltonian \(H(t)\), which combines control fields with intrinsic interactions and varies in time, dominating a temporal quantum state \(\ket{\psi(t)}\) via 
\begin{equation}
i \frac{d}{dt} \ket{\psi(t)} = H(t)\ket{\psi(t)}.
\end{equation}
Brute-force characterization scales poorly with system size~\cite{Mohseni2008Quantum}.
To mitigate this, parametric and prior-informed approaches have been developed, including time-series fitting~\cite{Di2009Hamiltonian,Zhang2014Quantum} and eigenstate-based observable reconstruction~\cite{Qi2019determininglocal,Eyal2019Learning,gu2024practical,haah2024learning}, as well as Bayesian inference~\cite{Wang2017Experimental}, machine-learning methods~\cite{Xin2018Local,an2024dual}, and cross-device strategies~\cite{wiebe2014hamiltonian,Chen2024A}.
Although highly effective for time-independent Hamiltonians, these methods face significant challenges when applied to driven systems~\cite{anshu2021sample,Yu2023robustefficient,Han2021Tomography,an2024dual}, typically failing to capture rapid temporal modulation, since most do not explicitly leverage periodicity.

We introduce a Floquet-informed learning algorithm that efficiently reconstructs an unknown time-periodic Hamiltonian on a prepared Floquet eigenstate. Representing \(H(t)\) by a truncated Fourier series and invoking the Floquet theorem~\cite{tsuji2023floquet,Barone1977Floquet}, we recast parameter estimation as a compact linear inverse problem in the Floquet band picture. The procedure achieves polynomial scaling in the initial states required and classical post-processing, and quantify the amount of restrictions required if errors exists. 
The method systematically addresses a broad range of driven dynamics, from explicit truncated Fourier series to those without a clear cutoff via an adaptive truncation selection, and extends to near-periodic drives with decaying high-frequency content. These capabilities position the framework as a practical tool for certifying and benchmarking driven quantum platforms, with immediate relevance to fault-tolerant control, topological-phase verification, and Floquet engineering.

\begin{figure}[t]
  \centering
  \includegraphics[width=0.9\linewidth]{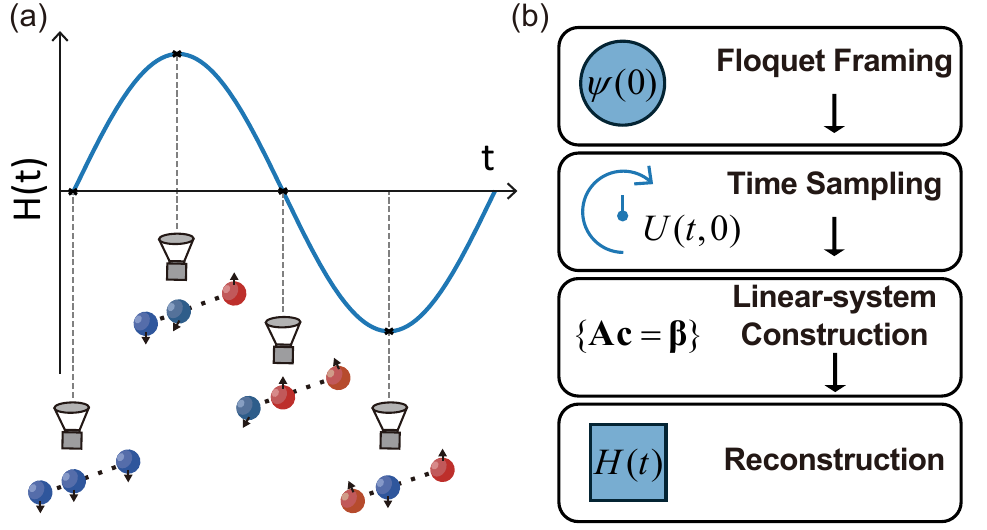}
  \caption{(a) Learning protocol for a time-dependent Hamiltonian \(H(t)\): a periodic drive (blue) is sampled at times \(t_n\), and system is initialized at a Floquet eigenstate. (b) Four-step algorithm: floquet framing, time sampling, linear-system construction, and reconstruction.}
  \label{fig:HL}
\end{figure}

\paragraph*{Algorithm.}
We focus on a time-periodic Hamiltonian with period \(T\) and cutoff \(M\) that can be expressed as,
\begin{eqnarray}\label{eq:H_fre}
    &&H(t)=\sum_{m=-M}^{M} e^{-im\omega t}\,H_m,\quad \omega=\frac{2\pi}{T},\nonumber \\
    &&H_m=\sum_{i=1}^{R} c_{m,i}\,P_{m,i},\quad P_{m,i}\in\mathbb{H}_{l},\ c_{m,i}\in\mathbb{R}.
\end{eqnarray}
Here, \(R\) is the number of chosen local bases per band and $P_{m,i}$ acts on a subspace with dimension \(l=\mathrm{poly}(\log d)\). This locality is physically motivated and crucial for scalability~\cite{Seth1996Universal}. 

Previous work~\cite{Eyal2019Learning} showed that static Hamiltonian (\(M=0\)) can be reconstructed from local observables with resources scaling linearly in system size.
We extend the scenario, using correlators of local observables,  to a general time-periodic Hamiltonian, i.e., \(M>0\). 
The algorithm is summarized in Fig.~\ref{fig:HL}(b), with main steps outlined below and details specified in Appendix~\ref{app:recon}.

\textit{1.Floquet framing.}
Denote \(U(T,0)\) the evolution operator over one period and Floquet Hamiltonian \(H_F\), an effective Hamiltonian defined as 
\[
  e^{-iH_F T}=U(T,0),
\]
an eigenstate \(\ket{\psi(0)}\) of \(H_F\) with any quasienergy \(\varepsilon_\alpha\) is prepared[See Appendix \ref{app:sub}]. The \(T\)-periodic Floquet mode is labeled as
\[
  \ket{u_\alpha(t)} = U(t,0)\,e^{iH_F t}\ket{\psi(0)}\Rightarrow \ket{\psi(t)}=e^{-i\varepsilon_\alpha t}\ket{u_\alpha(t)} .
\]

\textit{2.Time sampling.}
Sample \(\ket{u_\alpha(t_n)}\) at \(N\) equidistant times, \(t_n=nT/N\), with \(N\ge 2M{+}1\) ensuring Nyquist–Shannon coverage of all Fourier modes \(|m|\le M\), avoiding aliasing. 
All discrete Fourier components can be obtained as, 
\begin{eqnarray}
    \label{eq:expand}
    \ket{u_\alpha^k}=\frac{1}{N}\sum_{n=0}^{N-1} e^{ik\omega t_n}\ket{u_\alpha(t_n)}
\end{eqnarray}
Upon substitution into the Schrodinger equation, this produces the band-picture relation
\begin{equation}
  (\varepsilon_\alpha+k\omega)\ket{u_\alpha^k}
  =\sum_{m=k-M}^{k+M} H_{k-m}\ket{u_\alpha^{m}}
  \label{eq:sebp_unique}
\end{equation}
for integers  \(k\). Acting with a set of local observables \(\{A_j\}_{j=1}^{L}\subset\mathbb{H}_l\) gives linear identities
\begin{equation}
  \beta^\alpha_{k,j}=\sum_{m=k-M}^{k+M}\sum_{i=1}^{R} c_{k-m,i}\,a^{k,k-m}_{j,i},
  \label{eq:single_master}
\end{equation} 
where the correlators to be measured are
\[
  \beta^\alpha_{k,j}=(\varepsilon_\alpha+k\omega)\langle u_\alpha^k|A_j|u_\alpha^k\rangle,
  a^{k,k-m}_{j,i}=\langle u_\alpha^k|A_j P_{k-m,i}|u_\alpha^{m}\rangle.
\]

\textit{3.Linear system construction.}
Via varying \(k\) and \(A_j\), we collect sufficient correlators via Eq.~\eqref{eq:expand} on sampled \(\ket{u_\alpha(t)}\) and obtain an over-determined system \(\mathbf A_{(k)}\,\mathbf c=\boldsymbol\beta_{(k)}\) for the unknown coefficients, specified as
\footnotesize
\begin{equation}
\underbrace{\begin{pmatrix}
a^{k, M}_{1,1} & \cdots & a^{k, M}_{1,R} & \cdots & a^{k,-M}_{1,1} & \cdots & a^{k,-M}_{1,R} \\
\vdots & \ddots & \vdots & & \vdots & \ddots & \vdots \\
a^{k, M}_{L,1} & \cdots & a^{k, M}_{L,R} & \cdots & a^{k,-M}_{L,1} & \cdots & a^{k,-M}_{L,R}
\end{pmatrix}}_{\mathbf A_{(k)}\in\mathbb{C}^{L\times (2M+1)R}}
\underbrace{\begin{pmatrix}
c_{M,1}\\ \vdots\\ c_{M,R}\\ \vdots\\ c_{-M,1}\\ \vdots\\ c_{-M,R}
\end{pmatrix}}_{\mathbf c\in\mathbb{R}^{(2M+1)R}}
=
\underbrace{\begin{pmatrix}
\beta^\alpha_{k,1}\\ \vdots\\ \beta^\alpha_{k,L}
\end{pmatrix}}_{\boldsymbol\beta_{(k)}\in\mathbb{C}^{L}} .
\label{eq:master_block}
\end{equation}
\normalsize
For each \(k\), \(L\) linearly independent equations can be generated.
Choosing $L\ge R$ ensures that stacking \(k\) for \(\ge 2M{+}1\) can yield an over-determined system. 

\textit{4. Reconstruction}
Solving \(\mathbf A_{(k)}\,\mathbf c=\boldsymbol\beta_{(k)}\) for chosen \(k\) (e.g., via least squares) yields \(\mathbf c\). Hence the full Hamiltonian \(H(t)\) can be constructed via Eq.~\eqref{eq:H_fre}.

\paragraph*{Assumptions.}
\(\ket{\psi(0)}\) can be prepared via a variant quantum method and \(\varepsilon_\alpha\) can be determined by phase estimation algorithm. 
For correlator estimation, \(\beta^\alpha_{k,j}\), \(a^{k,k-m}_{j,i}\) are assembled from time-domain expectation values using Hadamard tests between the sampled \(\ket{u_\alpha(t)}\).  Remarkably, entire reconstruction algorithm require thus controlled evolutions operation, we can replace it by a two-copy SWAP-based strategy. The detailed information are all listed in Appendix~\ref{app:sub}.

\paragraph*{Complexity.}
Appendix~\ref{app:com} contains a full complexity analysis. Here we give a brief summary.
Each correlator required \(\mathcal{O}(N^2)\) times measurements. By stacking \(k\) (with number of \(k\), \(\ge 2M+1\)), the sample complexity scales as
\(
  \mathcal{O}\!\big(N^2\, (2M{+}1)^2\,L\,R\big).
\)
For fixed time resolution with \(N\propto T\), this becomes 
\(
\mathcal{O}\!\big(T^2 M^2\,\mathrm{poly}(\log d)\big).
\)
Post-processing solves a \(((2M{+}1)L)\times((2M{+}1)R)\) least-squares problem in time
\(
\mathcal{O}\!\big(M^3\,\mathrm{poly}(\log d)\big),
\)
typically dominated by data acquisition for moderate \(M\) and \(\log d\).

\paragraph*{Numerical Experiments.}
We assess the algorithm’s performance via simulations of time-periodic quantum systems. Specifically, we realize Eq.~\eqref{eq:H_fre} with two canonical many-body models, choosing each Fourier component \( H_m \) to be either an Ising or a Heisenberg Hamiltonian.
For the Ising model,
\begin{equation} \label{H_ising}
H_m = \sum_{\langle i,j \rangle} J^{(m)}_{ij} \sigma_z^i \sigma_z^j + \sum_i h^{(m)}_i \sigma_x^i,   
\end{equation}
where \( \sigma^i_{x,z} \) are Pauli matrices, \( J^{(m)}_{ij} \) the nearest-neighbor interaction strengths, and \( h^{(m)}_i \) transverse fields.
For the Heisenberg model, the Hamiltonian is given by
\begin{equation}
H_m = \sum_{k=x,y,z} \sum_{\langle i,j \rangle} J^{(m)}_{ij,k} \sigma_k^i \sigma_k^j + \sum_i h^{(m)}_i \sigma_x^i,
\end{equation}
where \( J^{(m)}_{ij,k} \) encodes direction-dependent couplings along the \(k\)-axis.
As illustrated in Fig.~\ref{fig:simu}(a), we implement the models on 1D chains and 2D lattices. The notation \(\langle i,j \rangle\) indicates nearest-neighbour sites, and periodic boundary conditions are imposed.

The simulations follow the procedure described above and summarized in Table~\ref{model_steps} of Appendix~\ref{app:recon}. Given a time period \(T\), we initialize the system in one eigenstate of the Floquet Hamiltonian, labeled as \(\ket{\psi(0)}\). Time evolution under the time-dependent Hamiltonian \(H(t)\) is simulated by discretizing the period into small steps \(\Delta t\) and computing the evolution operator via matrix exponentials \( \exp(-i H(t) \Delta t) \). The resulting states \(\ket{\psi(t)}\) are used to construct matrix elements of the form \( \bra{u_\alpha(j\Delta t)} A \ket{u_\alpha(j'\Delta t)} \) for \( j, j' \in \{1, \dots, N\} \), which form the correlator and enter the linear system \( \mathbf{A} \mathbf{c} = \boldsymbol{\beta} \). Finally, equations are solved via minimizing least squares, and the reconstructed Hamiltonian is then benchmarked against the target. 

\begin{figure}[ht]
    \centering    
    \includegraphics[width=1\linewidth]{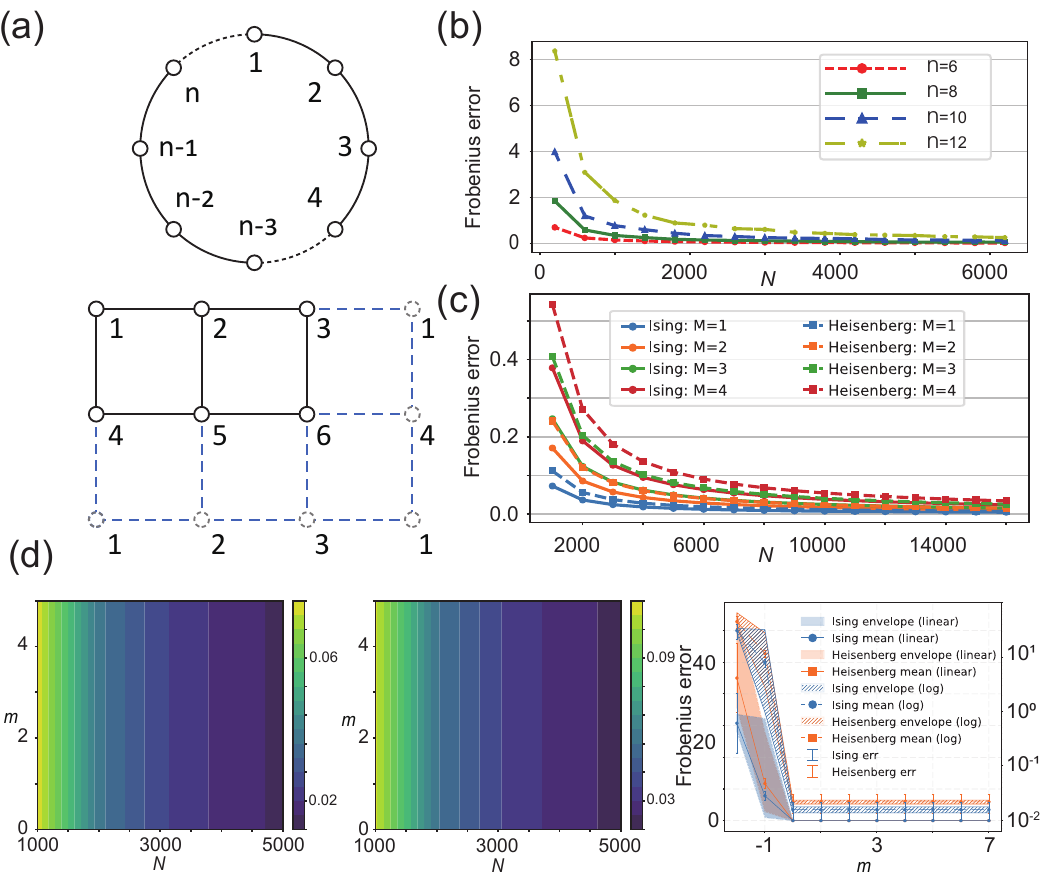}
    \caption{Results of the simulations. 
    (a) is the structure of simulated models. For the 1D model: (b) shows the Frobenius error as a function of the number of discretized time segments \( N \) to reconstruct an Ising Hamiltonian. 
    For the \(2 \times 3\) 2D model: (c) demonstrates the effect of varying the number of Fourier components \( M \) for both Ising and Heisenberg Hamiltonians. (d) presents the impact of varying the number of constraint equations. Error-bars are from averaging standard derivation when calculating Frobenius error at different time.}
    \label{fig:simu}
\end{figure}

In the first set of simulations, each Fourier component \(H_m\) is chosen to be an \(n\)-qubit 1D Ising chain with \(n\in \{6,8,10,12\}\). we restrict the drive to a single Fourier component (\(M=1\)).
The interaction parameters \( J_{ij}^{(m)} \) and transverse fields \( h_i^{(m)} \) are randomly initialized and recorded as target Hamiltonian.
After reconstruction we evaluate the Frobenius norm of the difference between the learned and target Hamiltonian, labeled as Frobenius error.
Fig.~\ref{fig:simu}(b) report the average over the 20 time instants, showing that the norm enjoys an exponentially decay as \(N\) increases across different system sizes.

We then consider \(H_m\) on a \(2\times3\) 2D lattice and study how different truncation order \(M\) affects performance for both Ising and Heisenberg Hamiltonians.
The interaction strengths \( J_{ij,k}^{(m)} \) and transverse fields \( h_i^{(m)} \) drawn randomly for each \( m \in [M] \) and serve as the target Hamiltonian.
By varying \( M = 1, 2, 3, 4 \)), we evaluate reconstruction accuracy using the Frobenius error, shown in Fig.~\ref{fig:simu}(c). Across different \(M\), the error decreases roughly exponentially with \(N\), indicating that the algorithm performs well provided the time discretization is sufficiently fine.

Finally, we examine the effect of stacking Eq.\eqref{eq:master_block} on reconstruction performance. In the original formulation the Fourier index runs over \(k\in[-M,M]\). We extend this to \(k\in[-M-m,M+m]\)to control the number of constraint equations. 
Focusing on the second model with \(M=1\), we first increase \(m\) from \(0\) to \(5\), finding that reconstructions converge and remain accurate for the tested time discretizations (left two panels of Fig.\ref{fig:simu}(d)). We then fix \(N=5000\), sweep  \(m\) from \(-2\) to \(7\), and generate 10 random instances. The rightmost panel in Fig.~\ref{fig:simu}(d) displays the mean Frobenius error across those instances and their envelope, which are additionally displayed on a logarithmic scale. The data indicate that, once the number of equations exceeds \( 2M+1 \), further increases do not yield improved performance.

\paragraph*{Applicability to More General Cases.}
Up to this point the Hamiltonian to be learned is strictly time-periodic with a known truncation order \(M\). In practice, however, \(M\) may be unknown or the Hamiltonian may not be strictly periodic. We therefore consider two more general scenarios. Scenario 1, the Hamiltonian is periodic, but either \(M\) is unknown or additional Fourier components exist beyond the chosen truncation. Scenario 2, the Hamiltonian has no intrinsic period (formally the limit \(T \to \infty\)), so that an infinite Fourier series is required; in this case we assume that high-frequency components are present but their amplitudes decay rapidly.

These situations can be summarized by the generalized form
\begin{equation}
\label{eq:H_general}
H(t) = \sum_{m=-M}^{M} e^{-i m \omega t} H_{m} + \varepsilon_0 \sum_{|m|>M} e^{-i m \omega t} H_{m}.
\end{equation}
The first sum represents the truncated Fourier series that our method will learn, while the second sum represents the neglected higher-order terms that could degrade the algorithm’s accuracy. The parameter \( \varepsilon_0 \) controls the relative contribution of this second term. 

For Scenario 1, Eq.~\eqref{eq:H_general} directly applies with a finite base frequency \( \omega = 2\pi/T\). 
For Scenario~2, we can interpret Eq.~\eqref{eq:H_general} in the limit of an infinitesimal base frequency \( \omega \) (formally \(T \to \infty\)). The first sum (with \(|m|\le M\)) picks up the dominant, low-frequency part of the Hamiltonian, whereas the second sum represents the remaining high-frequency components.  Under the assumption that these high-frequency components decay sufficiently fast, \(H(t)\) behaves as a near-periodic function.

We next benchmark the algorithm on Eq.~\eqref{eq:H_general} with a \(6\)-qubit 1D Ising Hamiltonian.
In the first test, we fix \( M = 5, 7, 9 \) and vary the length of the residual sum from \(0\) to \(9\), with \(\varepsilon_0 = 0.005, 0.01, 0.05\). For each configuration, we generate 10 random Hamiltonians and report the average Frobenius error.
As shown in Fig.~\ref{fig:simu2}(a), when \(\varepsilon_0\) is sufficiently small, the truncated model reproduces the true Hamiltonian accurately, whereas larger \(\varepsilon_0\) introduces significant high-frequency contributions and degrades performance.
In the second test, we fix the total number of Fourier terms in the Hamiltonian to 5, 7, or 9 and vary the algorithmic cutoff \(M\) from \(2\) to \(11\). 
The results in Fig.~\ref{fig:simu2}(b) demonstrate that once $M$ matches or exceeds the true order, reconstruction is accurate. Importantly, overspecifying $M$ does not harm performance, suggesting a robust foundation for adaptive truncation strategies.

When \(M\) is unknown or the Hamiltonian is near-periodic and has decaying high-order content, we use the following adaptive truncation rule: initialize with a small cutoff \(M\); increase \(M\) stepwise and refit the Hamiltonian each time; compute the discrepancy between successive reconstructions; stop when discrepancy falls below a preset threshold.
For strictly periodic \(H(t)\), the process stops once \(M \ge M_{\text{true}}\). For near-periodic \(H(t)\), the process is supposed to stop when adding Fourier modes yields negligible improvement, indicating that all physically relevant dynamics have been captured.

Intuitively, the smoother the drive is in time, the faster its Fourier weights decay.
If $H(t)$ is $T$-periodic and has $p$ continuous derivatives, the $(p\!+\!1)$-st derivative being integrable, then integration by parts shows $\|H_m\|$ falls off like $|m|^{-(p+1)}$~\cite{Trefethen2019Approximation}.  
Truncating at order $M$ then incurs an algebraic error $\|H(t)-H^{(M)}(t)\| = O(M^{-p})$. If, instead, $H(t)$ extends analytically into a strip of the complex plane and remains bounded there, contour shifting implies $\|H_m\|$ decays geometrically, giving an exponential tail $\|H(t)-H^{(M)}(t)\| = O(e^{-\gamma M})$ for some $\gamma>0$.
Thus, if $H(t)$ is time-periodic and smooth or near-periodic with decaying weights, choosing \(M\) sufficiently large ensures that the omitted term in our learning framework is small, and the algorithm can recover the dominant dynamics, supporting our adaptive strategy.

\begin{figure}[ht]
    \centering    
    \includegraphics[width=1\linewidth]{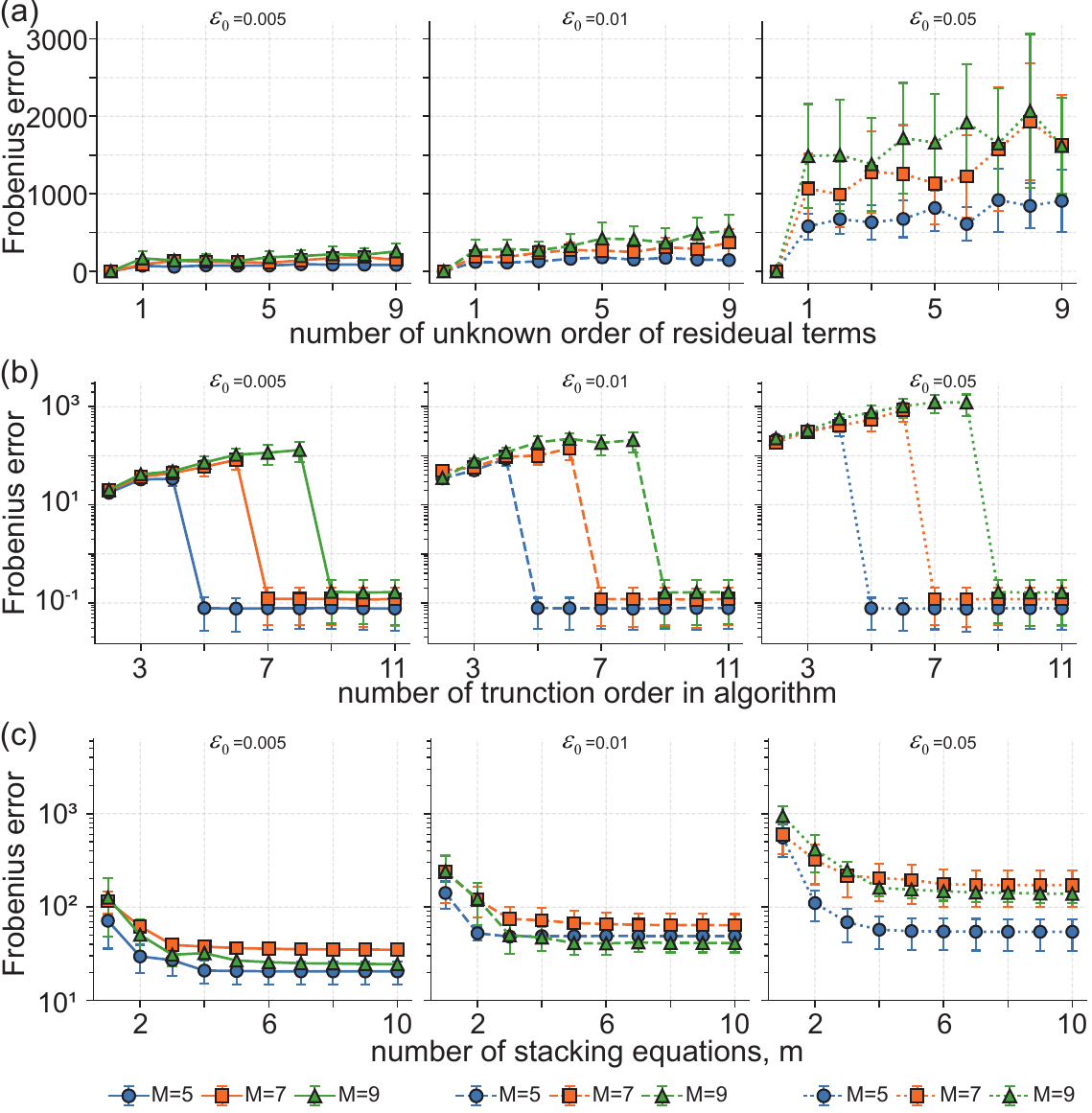}
    \caption{Results of the simulations. 
    (a) presents the Frobenius error if the learned order is fixed while the order of residual terms of Hamiltonian is varying(from 0 to 9). (b) presents the Frobenius error when fixing the total order of learned Fourier components of Hamiltonian and varying the learned order of algorithm(from 2 to 11). (c) presents the Frobenius error with \(m\) from 1 to 10. Error-bars are from averaging standard derivation when calculating Frobenius error at different time. }
    \label{fig:simu2}
\end{figure}

\paragraph*{Constraint Counts.}
In the noiseless, uniqueness of the solution is equivalent to $\mathbf{A}$ having full column rank. By enforcing observable completeness, sufficient stacking equations, i.e., band coverage, and a generic measurement design, one obtains practical sufficient conditions that guarantee $\mathbf{A}$ is full rank. In Fig.~\ref{fig:simu}(d) we see that, increasing the number of stacked equations over \([-M-m,M+m]\) with $m\ge 0$, brings no further benefit. We therefore investigate how stacking additional Floquet-band equations affects reconstruction in the presence of noise. For $M\in\{5,7,9\}$ we vary the number of stacked equations with \(m\) from 1 to 10 and inject noise via the residual terms in Eq.\eqref{eq:H_general}, with $\varepsilon_0\in\{0.005,0.01,0.05\}$ and a single Fourier component. For each configuration we compute the Frobenius errors. The results, shown in Fig.~\ref{fig:simu2}(c), indicate that increasing the number of equations improves performance under noise.

The observation is supported by following results, which is specified in Appendix~\ref{app:theory}. 
Under mild design and noise assumptions, least-squares reconstruction from Eq.~\eqref{eq:master_block} enjoys the following bound.
Let $K=(2M{+}1)R$ be the number of unknown coefficients and $S$ the number of stacked equations.
With probability at least $1-2\delta$,
\begin{equation}
\label{eq:main_text_bound}
\|\widehat{\mathbf c}-\mathbf c^\star\|_2 \;\le\;
C\,\frac{\sigma}{\mu_0}\sqrt{\frac{\mu_1\,(K+\log(1/\delta))}{S}}
\;+\; C'\,\frac{\sqrt{\mu_1}}{\mu_0}\,\varepsilon_0,
\end{equation}
where $\sigma^2$ is the noise scale for measuring \(\mathbf{\beta}\), $\mu_0,\mu_1$ quantify the ellipticity of the design covariance for obtaining elements in \(\mathbf{A}\), and $\varepsilon_0$ captures the truncated Fourier tail in Eq.~\eqref{eq:H_general}.
Consequently, to attain accuracy $\varepsilon$ it suffices to take
\begin{equation}
\label{eq:main_text_sample}
S=\Theta\!\Bigl(\frac{\kappa^2\,\sigma^2}{\varepsilon^2}\,(K+\log(1/\delta))\Bigr), \quad \kappa^2:=\mu_1/\mu_0^2.
\end{equation}
Theoretical considerations and numerical evidence consistently show that an overdetermined linear system mitigates sensitivity to noise and improves the robustness of Hamiltonian reconstruction.

\paragraph*{Conclusion.}
We have introduced a Floquet-informed algorithm for reconstructing time-periodic Hamiltonians. The experimental interface is experimentally acceptable for current platforms: it requires (i) preparation of a single Floquet eigenstate (e.g., via variational routines or adiabatic ramps) and (ii) measurement of correlators using ancilla-assisted estimators (e.g., Hadamard tests). Complexity is dictated primarily by the number of unknown coefficients rather than the Hilbert-space dimension, making the approach amenable to local lattice models and scalable certification workflows. Considering the noises, numerical simulation exhibits robustness.

In contrast to piecewise-time discretization with static Hamiltonian learners, which ignores temporal structure and suffers from poor scaling on fine time grids, our method achieves comparable accuracy with substantially fewer measurements whenever the Fourier spectrum is concentrated at low orders. Our default implementation relies on controlled simulation, but a two-copy, SWAP-based scheme can also be used. The selection between them can be made pragmatically, depending on the available experimental resources.

In the future, to stabilize the reconstruction in large-\(M\) regimes and circumvent the need to prespecify the cutoff, adaptive truncation may be combined with regularization. ridge (\(\ell_2\)) penalties damp high-frequency coefficients while LASSO (\(\ell_1\)) promotes sparse harmonic selection, yielding an effective bias–variance trade-off and improved noise robustness.
Future work includes automated identification of the base frequency, extension to strongly non–band-limited drives, and derivation of minimax finite-sample lower bounds in the presence of experimental noise. Achieving these goals would establish a standardized, platform, agnostic pathway for certifying driven quantum dynamics at scale, an enabling step toward precise quantum control and simulation.

\bibliographystyle{apsrev4-1}
\bibliography{floquent}

\onecolumngrid
\appendix
\newpage
\setcounter{figure}{0}
\renewcommand\thefigure{S\arabic{figure}}
\setcounter{table}{0}
\renewcommand\thetable{S\arabic{table}}

\section{Supplementary Information on Reconstruction Algorithm}
\label{app:recon}
We now provide the detailed derivation of the method for reconstructing a time-periodic Hamiltonian using its Fourier structure.  
By Floquet’s theorem, the evolution of a periodically driven quantum system can be described by a unitary operator over one period,
\begin{eqnarray}
    U(s+T,s) = \exp(-i H_F(s) T),
\end{eqnarray}
where \( H_F(s) \) is the Floquet Hamiltonian, a time-dependent but \(T\)-periodic operator, whose eigenvalues are the quasienergies.  
For evolution under a time-periodic Hamiltonian, the propagator can be expressed as 
\begin{eqnarray}
    U(t,t_0) = P_s(t) \exp\left[-i H_F(s) (t - t_0)\right] P_s(t_0)^{\dagger},
\end{eqnarray}
where \( P_s(t) = U(t,s) \exp[-i H_F(s) (s - t)] \) captures the periodic micromotion.

Without loss of generality, we set \(s = t_0 = 0\), getting a Floquet Hamiltonian shown in main text 
\[
U(T,0) = \exp(-i H_F(0) T),\Rightarrow H_F \equiv H_F(0),
\]
and choose \( \ket{\psi(0)} \) as the initial state of the system.  
If \( \ket{\psi(0)} \) is an eigenstate of \( H_F \), the state at time \( t \) evolves as
\begin{eqnarray} \label{time_st}
  \ket{\psi(t)}=U(t,0)\ket{\psi(0)}=e^{-i \varepsilon_\alpha t}P_0(t) \ket{\psi(0)} = e^{-i \varepsilon_\alpha t} \ket{u_\alpha(t)},
\end{eqnarray}
where \( \ket{u_\alpha(t)} = P_0(t) \ket{\psi(0)} \) is periodic in \( t \), and \( \varepsilon_\alpha \) is the corresponding quasienergy.  
The periodic function \( \ket{u_\alpha(t)} \) admits a Fourier expansion,
\begin{eqnarray} \label{psi_fre}
     \ket{u_\alpha(t)} = \sum_{m} e^{-im \omega t} \ket{u_\alpha^m},  
     \quad 
     \ket{u_\alpha^m} = \frac{1}{T}\int_0^T e^{im \omega t} \ket{u_\alpha(t)}\, dt,
\end{eqnarray}
where \( \ket{u_\alpha^m} \) is the Fourier component at frequency \( m \omega \) with \( \omega \) the driven frequency. In our simulation, we set it as \(4\pi\).

Substituting Eqs.~\eqref{eq:H_fre}, \eqref{psi_fre}, and \eqref{time_st} into the Schr\"odinger equation gives
\begin{eqnarray}
 i\frac{d}{dt} \left( e^{-i \varepsilon_\alpha t} \sum_{m} e^{-im \omega t} \ket{u_\alpha^m} \right) =
 \left( \sum_{m'} e^{-im' \omega t} H_{m'} \right) 
 \left( e^{-i \varepsilon_\alpha t} \sum_{m} e^{-im \omega t} \ket{u_\alpha^m} \right).
\end{eqnarray}
Matching coefficients of \( e^{-ik\omega t} \) yields the stationary (banded) equations:
\begin{eqnarray}
   (\varepsilon_\alpha + k \omega)\, \ket{u_\alpha^k} = \sum_{m = k-M}^{k+M}  H_{k-m} \ket{u_\alpha^m},
\end{eqnarray}
where \( \ket{u_\alpha^k} \) is computed from discrete sampling as
\[
\ket{u_\alpha^k} = \frac{1}{N} \sum_{n=1}^{N} e^{ik \omega n T/N} \ket{u_\alpha(nT/N)},
\]
with \( N \) the number of time samples.

Let \( \{A_{j}\} \) be a set of observables supported on the local Hilbert space \(\mathbb{H}_l\).  
Applying \(A_j\) to both sides of the \(k\)-th equation and taking the expectation value gives
\begin{equation} \label{eq:pro2}
(\varepsilon_\alpha + k\omega) \langle u_\alpha^k | A_{j} | u_\alpha^k \rangle
= \sum_{m = k-M}^{k+M} \langle u_\alpha^k | A_{j} H_{{k-m}} | u_\alpha^m \rangle.
\end{equation}
Using the local Hamiltonian expansion
\begin{equation} \label{eq:local_ham}
H_m = \sum_{i=1}^{R} c_{m, i} P_i,
\end{equation}
Eq.~\eqref{eq:pro2} becomes
\begin{equation} \label{eq:pro2_v2}
\beta^{\alpha}_{k,j} = \sum_{m = k-M}^{k+M} \sum_{i=1}^{R} c_{k-m, i} \, a^{k, k-m}_{j,i},
\end{equation}
where
\[
\beta^{\alpha}_{k,j} = (\varepsilon_\alpha + k\omega) \langle u_\alpha^k | A_{j} | u_\alpha^k \rangle,
\quad
a^{k, k-m}_{j,i} = \langle u_\alpha^k | A_{j} P_{i} | u_\alpha^m \rangle.
\]

By varying \(k\) and \(A_j\) over the chosen ranges , we obtain the system
\begin{equation} \label{master_function}
\begin{pmatrix}
a^{k, M}_{0,0} & a^{k, M}_{0,1} & \dots & a^{k, -M}_{0,R} \\
a^{k, M}_{1,0} & a^{k, M}_{1,1} & \dots & a^{k, -M}_{1,R} \\
\vdots & \vdots & \ddots & \vdots \\
a^{k, M}_{L,0} & a^{k, M}_{L,1} & \dots & a^{k, -M}_{L,R}
\end{pmatrix}
\begin{pmatrix}
c_{M,0} \\
c_{M,1} \\
\vdots \\
c_{-M,R}
\end{pmatrix}
=
\begin{pmatrix}
\beta^{\alpha}_{k,0} \\
\beta^{\alpha}_{k,1} \\
\vdots \\
\beta^{\alpha}_{k,J}
\end{pmatrix}.
\end{equation}
Here, the number of unknowns is \(N_{\text{var}} = (2M + 1) R\).  
Choosing at least \(2M+1\) distinct \(k\) values ensures that the total number of equations \(N_{\text{eq}}\) exceeds \(N_{\text{var}}\), allowing a unique solution.

We write the compact form
\begin{equation} \label{eq:linear_system_matrix}
\mathbf{A} \mathbf{c} = \boldsymbol{\beta},
\end{equation}
where \(\mathbf{A} \in \mathbb{C}^{N_{\text{eq}}\times N_{\text{var}}}\) has entries \(a^{k,k-m}_{j,i}\),  
\(\mathbf{c}\) collects all \(c_{m,i}\), and \(\boldsymbol{\beta}\) contains all \(\beta^{\alpha}_{k,j}\).  
Solving Eq.~\eqref{eq:linear_system_matrix} yields the Fourier components \(H_m\) of the target Hamiltonian.

\floatname{algorithm}{Table}
\begin{figure}
\begin{algorithm}[H]
\caption{Hamiltonian Learning Algorithm}
\label{model_steps}
\begin{algorithmic}[1]
    \State \textbf{Input}: Unknown time-periodic Hamiltonian \( H(t) \) with known period \( T \) and bounded \( M \), such that
    \[
    H(t) = \sum_{m=-M}^{M} e^{-im2\pi t/T} H_m
    \]    
    \State \textbf{Step 1: Floquet framing}
    \State Prepare an eigenstate of the Floquet Hamiltonian \( H_F \), denoted as \( \ket{\psi(0)} \).
    
    \State \textbf{Step 2: time sampling}
    \State Discretize the time interval \( T \) into \( N = T/\Delta t \) segments.
    \State Design a scattering circuit (see Figure~\ref{scattering_cir}) to obtain the matrix elements \( \bra{u_{\alpha}(j \Delta t)} A \ket{u_{\alpha}(j' \Delta t)} \) for \( j, j' \in \{1, 2, \dots, N\} \).
    
    \State \textbf{Step 3: Linear system construction}
    \State Acquire sufficient matrix elements \( \bra{u^{n'}_{\alpha}} A \ket{u^{n}_{\alpha}} \) by varying \( k \) from \(-M-1 \) to \( M+1 \) and \( \{A\} \). 
    \State Construct the linear system of equations \( \mathbf{A} \mathbf{c} = \boldsymbol{\beta} \) as described in Eq.~\eqref{master_function}.
    
    \State \textbf{Step 4: Reconstruction}
    \State Verify that the coefficient matrix \( \mathbf{A} \) has sufficient independent equations to ensure a unique solution.
    \State Solve for the coefficients \( c_{k,i} \) using a least-squares approach.
    \State Reconstruct the Hamiltonian \( H(t) \) using the solved coefficients \( c_{k,i} \).
    
    \State \textbf{Output}: Reconstructed Hamiltonian coefficients \( \{c_{k,i}\} \) and a comparison to the target Hamiltonian.
\end{algorithmic}
\end{algorithm}
\end{figure}

\section{Supplementary Information on Subroutines}
\label{app:sub}

\subsection{Preparation of Eigenstate of $H_F$}

Algorithm for initialization is employed to prepare an initial state $\ket{\psi(0)}$, which is an eigenstate of the Floquet Hamiltonian $H_F$.  
This is implemented using a parameterized encoding circuit $U(\bm{\theta})$ and its conjugate decoding circuit $U^{\dagger}(\bm{\theta})$, as shown in Fig.~\ref{fig:steady_state}.

The algorithm assumes that $U(T,0)$ can be queried repeatedly and that the encoding circuit is parameterized by $U(\bm{\theta})$, which include the expression of target state.  
We formulate the problem as an optimization whose objective is to maximize the quantity
\begin{eqnarray}
    \frac{\|\bra{0}\, U^{\dagger}(\bm{\theta}) \, U(T,0) \, U(\bm{\theta}) \, \ket{0}\|}{P},
\end{eqnarray} 
where $P$ is a factor given as probability measured in circuit of Fig.~\ref{fig:factor}, 
\[
P=\bra{0}\, U^{\dagger}(\bm{\theta}) \,U^{\dagger}(T,0) U(T,0) \, U(\bm{\theta}) \, \ket{0}
\]
and \(\bm{\theta}\) are parameters to be learned. In our case, as \(H(t)\) keeps Hermitian, the entire process is unitary and \(P=1\). 
Figure~\ref{fig:steady_state} illustrates this procedure, where the probability of measuring $\ket{0}$ serves as the cost function.  
Various hybrid optimization protocols have been investigated for this task, and we can adopt an iterative algorithm.  
In each iteration, the query complexity scales proportionally with the number of parameters.  
The total complexity is difficult to predict, since the number of iterations required for convergence is problem-dependent, which remains an open question.

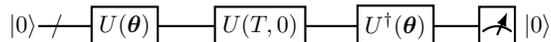
\begin{figure}[!h]
   \centerline{
   \begin{tikzpicture}[thick]
   \ctikzset{scale=1.8}
   \tikzstyle{every node}=[font=\normalsize,scale=0.9]
     \tikzstyle{operator} = [draw,shape=rectangle
     ,fill=white,minimum width=1em, minimum height=1em] 
     \tikzstyle{operator2} = [draw,shape=rectangle,fill=white,minimum width=3em, minimum height=9.5em] 
     \tikzstyle{operator22} = [draw,shape=rectangle,fill=white,minimum width=3em, minimum height=9.5em] 
     \tikzstyle{operator3} = [draw,shape=rectangle,fill=white,minimum width=3em, minimum height=1em] 
     \tikzstyle{operator4} = [draw,shape=rectangle,dashed, minimum width=1.5cm, minimum height=1cm] 
     \tikzstyle{operator5} = [draw,shape=rectangle,dashed, minimum width=5.75cm, minimum height=3cm] 
     \tikzstyle{operator6} = [draw=pink,shape=rectangle,dashed, minimum width=5cm, minimum height=4cm] 
     \tikzstyle{phase} = [fill,shape=circle,minimum size=3pt,inner sep=0pt]
     \tikzstyle{surround} = [fill=blue!10,thick,draw=black,rounded corners=2mm]
     \tikzstyle{ellipsis} = [fill,shape=circle,minimum size=2pt,inner sep=0pt]
     \tikzstyle{ellipsis1} = [draw,shape=circle,minimum size=2pt,inner sep=0pt]
     \tikzstyle{ellipsis3} = [draw,shape=circle,fill=white, minimum size=2pt,inner sep=1pt]
     \tikzstyle{ellipsis4} = [draw,shape=circle,fill=white, minimum size=4pt,inner sep=3pt]
     \tikzstyle{ellipsis2} = [draw,shape=circle,minimum size=0.5pt,inner sep=0pt]
     \tikzset{meter/.append style={fill=white, draw, inner sep=5, rectangle, font=\vphantom{A}, minimum width=15, 
     path picture={\draw[black] ([shift={(.05,.2)}]path picture bounding box.south west) to[bend left=40] ([shift={(-.05,.2)}]path picture bounding box.south east);\draw[black,-latex] ([shift={(0,.15)}]path picture bounding box.south) -- ([shift={(.15,-.08)}]path picture bounding box.north);}}}
     \node at (0.5,-1.5) (qin){$\ket{0}$};
     \node at (0.55,-1.5) (q3) {};
     \node[] (end3) at (4,-1.5) {} edge [-] (q3);
     \node at (0.75,-1.5) (qin){$/$};
     \node[operator] (op22) at (1.25,-1.5) {$U(\bm{\theta})$} ;
     \node[operator] (op22) at (2.25,-1.5) {$U(T,0)$} ;
     \node[operator] (op22) at (3.25,-1.5) {$U^{\dagger}(\bm{\theta})$} ;
     \node[meter] at (4,-1.5) (qin){};
     \node at (4.3,-1.5) (q3) {$\ket{0}$};
   \end{tikzpicture} } 
   \caption{Circuit for optimizing a steady state (Floquet eigenstate). The cost function is the probability of measuring $\ket{0}$ at the output.}   
\label{fig:steady_state} 
\end{figure}

\begin{figure}[!h]
   \centerline{
   \begin{tikzpicture}[thick]
   \ctikzset{scale=1.8}
   \tikzstyle{every node}=[font=\normalsize,scale=0.9]
     \tikzstyle{operator} = [draw,shape=rectangle
     ,fill=white,minimum width=1em, minimum height=1em] 
     \tikzstyle{operator2} = [draw,shape=rectangle,fill=white,minimum width=3em, minimum height=9.5em] 
     \tikzstyle{operator22} = [draw,shape=rectangle,fill=white,minimum width=3em, minimum height=9.5em] 
     \tikzstyle{operator3} = [draw,shape=rectangle,fill=white,minimum width=3em, minimum height=1em] 
     \tikzstyle{operator4} = [draw,shape=rectangle,dashed, minimum width=1.5cm, minimum height=1cm] 
     \tikzstyle{operator5} = [draw,shape=rectangle,dashed, minimum width=5.75cm, minimum height=3cm] 
     \tikzstyle{operator6} = [draw=pink,shape=rectangle,dashed, minimum width=5cm, minimum height=4cm] 
     \tikzstyle{phase} = [fill,shape=circle,minimum size=3pt,inner sep=0pt]
     \tikzstyle{surround} = [fill=blue!10,thick,draw=black,rounded corners=2mm]
     \tikzstyle{ellipsis} = [fill,shape=circle,minimum size=2pt,inner sep=0pt]
     \tikzstyle{ellipsis1} = [draw,shape=circle,minimum size=2pt,inner sep=0pt]
     \tikzstyle{ellipsis3} = [draw,shape=circle,fill=white, minimum size=2pt,inner sep=1pt]
     \tikzstyle{ellipsis4} = [draw,shape=circle,fill=white, minimum size=4pt,inner sep=3pt]
     \tikzstyle{ellipsis2} = [draw,shape=circle,minimum size=0.5pt,inner sep=0pt]
     \tikzset{meter/.append style={fill=white, draw, inner sep=5, rectangle, font=\vphantom{A}, minimum width=15, 
     path picture={\draw[black] ([shift={(.05,.2)}]path picture bounding box.south west) to[bend left=40] ([shift={(-.05,.2)}]path picture bounding box.south east);\draw[black,-latex] ([shift={(0,.15)}]path picture bounding box.south) -- ([shift={(.15,-.08)}]path picture bounding box.north);}}}
     \node at (0.25,-1.5) (qin){$\ket{0}$};
     \node at (0.3,-1.5) (q3) {};
     \node[] (end3) at (3,-1.5) {} edge [-] (q3);
     \node at (0.5,-1.5) (qin){$/$};
     \node[operator] (op22) at (1,-1.5) {$U(\bm{\theta})$} ;
     \node[operator] (op22) at (1.75,-1.5) {$U(T,0)$} ;
     \node[meter] at (3,-1.5) (qin){};
     \node at (3.3,-1.5) (q3) {$I$};
   \end{tikzpicture} } 
   \caption{Circuit for given a probability as factor.}   
\label{fig:factor} 
\end{figure}
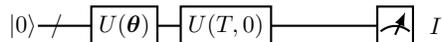

With respect to eigenphase (eigenvalue) extraction, if a controlled simulation of the time evolution \(U(t,0)\) is available, one may employ the Hadamard test or standard phase estimation.

If controlled simulation is unavailable, an alternative is to use two copies together with a SWAP operation. Let
\[
\rho = \ket{\psi(0)}\!\bra{\psi(0)},\qquad
\rho' = \ket{\psi(t)}\!\bra{\psi(t)},\qquad
\ket{\psi(t)} = U(t,0)\ket{\psi(0)}.
\]
For any observable \(A_r\) acting on the first copy, one has
\begin{equation}
\label{eq:swap-correlator}
\mathrm{Tr}\!\big[(A_r\otimes I)\,\mathrm{SWAP}\,(\rho\otimes\rho')\big]
= \bra{\psi(t)}A_r\ket{\psi(0)}\;\bra{\psi(0)}U(t,0)\ket{\psi(0)}.
\end{equation}
The left-hand side is experimentally accessible because \(\mathrm{SWAP}\) can be decomposed into a sum of measurable observables.
A convenient choice is \(A_r=\ket{0}\!\bra{0}\). In that case
\begin{equation}
\label{eq:swap-zeroprojector}
\mathrm{Tr}\!\big[(\ket{0}\!\bra{0}\otimes I)\,\mathrm{SWAP}\,(\rho\otimes\rho')\big]
= \braket{\psi(t)|0}\,\braket{0|\psi(0)}\;\braket{\psi(0)|U(t,0)|\psi(0)}.
\end{equation}
The overlaps \(\braket{\psi(t)|0}\) and \(\braket{0|\psi(0)}\) can be estimated by projective measurements onto \(\ket{0}\) with reference phase as 0. Hence \(\braket{\psi(0)|U(t,0)|\psi(0)}\) is obtained by dividing the measured SWAP correlator~\eqref{eq:swap-zeroprojector} by the independently estimated factor \(\braket{\psi(t)|0}\,\braket{0|\psi(0)}\). This only require a explict decomposition of SWAP matrix across two systems.

\subsection{Extraction of Correlators}
If if a controlled simulation of the time evolution \(U(t,0)\) is available, to obtain correlators of the form $\bra{u^{n'}_{\alpha}} A \ket{u^n_{\alpha}}$ with $n \neq n'$ in general, we employ the Hadamard test circuit.  
From Eq.~\eqref{psi_fre}, one can write
\begin{eqnarray} \label{observables}
\bra{u^{n'}_{\alpha}} A \ket{u^n_{\alpha}} 
= \sum_{j, j'} \bra{u_{\alpha}( j' \Delta t)} \, e^{-i n' \omega j' \Delta t} \, A \, e^{i n \omega j \Delta t} \, \ket{u_{\alpha}( j \Delta t)},
\end{eqnarray}
where $\ket{u_{\alpha}( j \Delta t)} = U(j\Delta t) \ket{\psi(0)}$.  

If the evolution period is discretized into $N_t$ time steps, there are in total $N_t^2$ such terms in Eq.~\eqref{observables}.  
We classify them into:

(a) \((j = j') \land (n = n')\), and  
(b) \((j \neq j') \lor (n \neq n')\).

For case (a), each term
\begin{eqnarray}
    \bra{u_{\alpha}( j \Delta t)}  A \ket{u_{\alpha}( j \Delta t)}, \quad j \in [N_t],
\end{eqnarray}
can be directly measured: evolve the system to $t = j \Delta t$ and measure observable $A$.

For case (b), we use the Hadamard test circuit in Fig.~\ref{scattering_cir}, where $l$ and $k$ are integers specifying the initial time index and the time separation, respectively.  
After the controlled-$A'$ operation and before measurement, the joint state is
\begin{eqnarray}
 \ket{0} \ket{u(l)} + e^{i\delta} \ket{1} A \ket{u(l+k)},
\end{eqnarray}
where $\ket{u(l)} \equiv \ket{u_{\alpha}(l \Delta t)}$ and $\delta = \omega \, (n l + n k - n' l) \, \Delta t$.  
Measuring the ancilla qubit yields
\begin{eqnarray}\label{real_p}
\langle \sigma_x \otimes I \rangle &=& \mathrm{Re} \left[ e^{i\delta} \bra{u(l)} A \ket{u(l+k)} \right],  \\ \label{imag_p}
\langle \sigma_y \otimes I \rangle &=& \mathrm{Im} \left[ e^{i\delta} \bra{u(l)} A \ket{u(l+k)} \right].
\end{eqnarray}  
Combining Eqs.~\eqref{real_p} and \eqref{imag_p} reconstructs the complex correlator for $(j \neq j') \lor (n \neq n')$ in Eq.~\eqref{observables}.

In summary, estimating a given $\bra{u^{n'}_{\alpha}} A \ket{u^n_{\alpha}}$ requires $N_t$ direct measurements for case (a) terms, and $\mathcal{O}(N_t^2)$ Hadamard-test circuits for case (b) terms.

\begin{figure}[!h]
   \centerline{
   \begin{tikzpicture}[thick]
   \ctikzset{scale=1.8}
   \tikzstyle{every node}=[font=\normalsize,scale=0.9]
     \tikzstyle{operator} = [draw,shape=rectangle
     ,fill=white,minimum width=1em, minimum height=1em] 
     \tikzstyle{operator2} = [draw,shape=rectangle,fill=white,minimum width=3em, minimum height=9.5em] 
     \tikzstyle{operator22} = [draw,shape=rectangle,fill=white,minimum width=3em, minimum height=9.5em] 
     \tikzstyle{operator3} = [draw,shape=rectangle,fill=white,minimum width=3em, minimum height=1em] 
     \tikzstyle{operator4} = [draw,shape=rectangle,dashed, minimum width=1.5cm, minimum height=1cm] 
     \tikzstyle{operator5} = [draw,shape=rectangle,dashed, minimum width=5.75cm, minimum height=3cm] 
     \tikzstyle{operator6} = [draw=pink,shape=rectangle,dashed, minimum width=5cm, minimum height=4cm] 
     \tikzstyle{phase} = [fill,shape=circle,minimum size=3pt,inner sep=0pt]
     \tikzstyle{surround} = [fill=blue!10,thick,draw=black,rounded corners=2mm]
     \tikzstyle{ellipsis} = [fill,shape=circle,minimum size=2pt,inner sep=0pt]
     \tikzstyle{ellipsis1} = [draw,shape=circle,minimum size=2pt,inner sep=0pt]
     \tikzstyle{ellipsis3} = [draw,shape=circle,fill=white, minimum size=2pt,inner sep=1pt]
     \tikzstyle{ellipsis4} = [draw,shape=circle,fill=white, minimum size=4pt,inner sep=3pt]
     \tikzstyle{ellipsis2} = [draw,shape=circle,minimum size=0.5pt,inner sep=0pt]
     \tikzset{meter/.append style={fill=white, draw, inner sep=5, rectangle, font=\vphantom{A}, minimum width=15, 
     path picture={\draw[black] ([shift={(.05,.2)}]path picture bounding box.south west) to[bend left=40] ([shift={(-.05,.2)}]path picture bounding box.south east);\draw[black,-latex] ([shift={(0,.15)}]path picture bounding box.south) -- ([shift={(.15,-.08)}]path picture bounding box.north);}}}
     \node at (0.5,-0.5) (qin){$\ket{0}$};
     \node at (0.55,-0.5) (q3) {};
     \node[] (end3) at (4,-0.5) {} edge [-] (q3);
     \node at (0.5,-1) (qin){$\ket{0}$};
     \node at (0.55,-1) (q3) {};
     \node[] (end3) at (4,-1) {} edge [-] (q3);
     \node at (0.7,-1.) (qin){$/$};
     \node[operator] (op22) at (1.25,-0.5) {$H$} ;
     \node[operator] (op22) at (1.25,-1.) {$U(0, l\Delta t)$} ;
     \node[ellipsis1] (op23) at (2.5,-0.5){};
     \node[] (end3) at (2.5,-1) {} edge [-] (op23);
     \node[operator] (op22) at (2.5,-1.) {$U(l\Delta t, (l+k)\Delta t)$};
     \node[ellipsis1] (op23) at (3.5,-0.5){};
     \node[] (end3) at (3.5,-1) {} edge [-] (op23);
     \node[operator] (op22) at (3.5,-1.) {$A'$};
     \node[meter] at (4,-0.5) (qin){};
   \end{tikzpicture} } 
   \caption{Hadamard test circuit for measuring case (b) correlators with $j \neq j'$ or $n \neq n'$. The controlled-$A'$ operation implements $A$ conditioned on the ancilla qubit.}   
\label{scattering_cir} 
\end{figure}
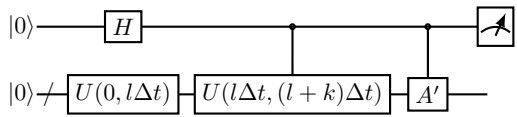

Alternatively, when controlled quantum simulation is unavailable we can adopt a two-copy SWAP-based strategy. Let
\[
\rho(t)=\ket{u(t)}\!\bra{u(t)},\qquad
\rho(t')=\ket{u(t')}\!\bra{u(t')}.
\]
For any observable \(A_r\) acting on the first copy one has
\[
\mathrm{Tr}\!\big[(A_r\otimes I)\,\mathrm{SWAP}\,(\rho(t)\otimes\rho(t'))\big]
= \langle u(t')|A_r|u(t)\rangle\;\langle u(t)|u(t')\rangle.
\]
The overlap \(s=\langle u(t)|u(t')\rangle\) can be estimated independently (e.g. via the eigenphase extraction method described above). Dividing the measured value by \(s\) then yields \(\langle u(t')|A_r|u(t)\rangle\). This approach removes the need for controlled evolution at the expense of preparing two copies and a modest increase in the number of experimental repetitions.

\section{Supplementary Information on Complexity}
\label{app:com}
We now analyze the sample and time complexities of our method, assuming ideal (noiseless) measurements.  
As described above, each measurement requires preparing an eigenstate \(\ket{\psi(0)}\) of the Floquet Hamiltonian \(H_F\).  
The central question is therefore: 
\emph{how many times must \(\ket{\psi(0)}\) be prepared}?

This value is determined by the requirement to estimate a sufficient set of matrix elements \(\bra{u_{\alpha}(j \Delta t)} A \ket{u_{\alpha}(j' \Delta t)}\), from which we construct \(\bra{u^{n'}_{\alpha}} A \ket{u^{n}_{\alpha}}\).  
From Eq.~\eqref{master_function}, the linear system involves \((2M + 1)R\) unknown parameters, where \(R \sim \mathrm{poly}(\log d)\) is the number of local basis operators.  
With numerical tests supporting, number of equations is larger than \((2M+1)R\), the solution can be obtained. We label the number of equation as \((2M+1)L\), where \(L\) labels stacking the equations. 
Thus, the number of distinct observable–Fourier-component pairs to be measured therefore scales as \((2M + 1)^2 RL\), with \(L\sim R\). 
Each such quantity can be obtained from \(\mathcal{O}(N^2)\) time-domain correlators, where \(N\) is the number of discrete time samples per period, \(T=N\Delta t\) with \(\Delta t\) the time resolution. Here we regard each measurement can be extracted accurately and count it as \(\mathcal{O}(1)\).

Consequently, the total number of state preparations required is
\[
\mathcal{O}\!\left({N^2 M^2} \,\mathrm{poly}(\log d)\right),
\]
where the factor \(\mathrm{poly}(\log d)\) approximates \(RL\).  
We refer to this as the sample complexity of our method (i.e., the number of quantum state preparations), not yet including classical post-processing.

Next, consider the total quantum evolution time.  
Each matrix element \(\bra{u_{\alpha}(j \Delta t)} A \ket{u_{\alpha}(j' \Delta t)}\) requires evolving the system for a duration of order \(T\), which is conducted via the Hadamard-test circuit or SWAP-based strategy. 
Thus, for \(\mathcal{O}(N^2 M^2\mathrm{poly}(\log d))\) measurements, the total evolution time scales as
\[
\mathcal{O}\!\left(N^2 M^2 \, T \,\mathrm{poly}(\log d)\right).
\]
From numerical experiments as a function of \(N\), the estimation error decays rapidly in all cases. 

Remarkably, although the SWAP-based strategy removes the need for controlled evolution, it can incur prohibitive measurement overheads: decomposing \(\mathrm{SWAP}^{\otimes n}\) into measurable Hermitian terms yields \(2^n\sim 4^{n}\) terms, i.e.\ exponential scaling in system size. Deploying this approach on larger systems therefore requires more efficient measurement schemes, especially for some specific designed \((A_r\otimes I)~\mbox{SWAP}\). When considering the cost, thus, we instead report the measurement overheads under the assumption that controlled simulation is available.

Finally, in the classical post-processing step we solve the \((2M + 1)L \times (2M + 1)R\) linear system from Eq.~\eqref{master_function}.  
Using standard dense linear-algebra methods (Gaussian elimination, LU decomposition, or least-squares solvers), this requires
\[
\mathcal{O}\!\big(M^3\,\mathrm{poly}(\log d)\big)
\]
time, since the matrix dimension is proportional to \(M \mathrm{poly}(\log d)\) and solving such a system scales cubically in its size.  
This post-processing cost is typically negligible compared to quantum measurement time unless \(M\) or \(\log d\) are very large; if the system matrix is sparse, more efficient solvers may be used.
In Table~\ref{tab:resources}, we summarize resources required for this algorithm considering a noiseless case.
\begin{table}[t]
\centering
\setlength{\tabcolsep}{6pt}
\renewcommand{\arraystretch}{1.15}
\caption{Resource scaling for reconstructing the Fourier components $\{H_m\}$ of a $T$-periodic Hamiltonian.  
Here $M$ is the truncation order, $R,L\sim \mathrm{poly}(\log d)$ is the number of local basis operators, $N$ is the number of time samples per period, and $\epsilon$ is the target additive precision per observable.}
\label{tab:resources}
\begin{tabular}{lll}
\hline\hline
\textbf{Resource} & \textbf{Scaling} & \textbf{Notes/Assumptions} \\
\hline
Eigenstate prep.\ of $H_F$ & -- & Via variant method; quasienergy by phase estimation/SWAP-based strategy \\
Distinct observables & $\Theta\!\big((2M{+}1)^2 RL\big)$ & Via Eq.~\eqref{master_function}\\
Correlators per observable & $\mathcal{O}(N^2)$ & Hadamard-test circuit\\
Sample complexity & $ \mathcal{O}\!\left(N^2 M^2 RL\right)$ &
Number of state preparations \\
Total evolution time & $ \mathcal{O}\!\left(N^2 M^2 T\, RL \right)$ &
Avg.\ evolution $\sim \Theta(T)$ per estimate \\
Post-processing time & $\mathcal{O}\!\big(((2M{+}1)RL)^3\big)$ &
Dense least-squares / LU decomposition\\
Additional Memory(quantum)   & -- & Via variant method \\
& $\mathcal{O}\!\big(\log (1/\epsilon))$ &  Quasienergy by phase estimation with $\epsilon$ accuracy \\
& 1 &  Correlators extraction via hadamard test \\
Memory (solver) & $\mathcal{O}\!\big(((2M{+}1)RL)^2\big)$ & Store normal equations / factors \\
\hline\hline
\end{tabular}
\end{table}

\section{Supplementary Information on High-probability Sufficient Constraint Counts}
\label{app:theory}
\paragraph*{Ideal-case uniqueness}
In the absence of noise and model truncation, the equations in $\mathbf A\mathbf c=\boldsymbol\beta$ are exactly consistent with the physical model, and uniqueness reduces to verifying that $\mathbf A$ has full column rank, $\mathrm{rank}(\mathbf A)=(2M{+}1)R$.  
In realistic experimental or numerical settings, however, statistical noise and truncation errors generally render the system inconsistent, and reconstruction is performed via a least-squares fit; in this case, we discuss the robustness to the true solution.

In the ideal case, a practical set operations to guarantee the full column rank are:
\begin{enumerate}
\item Observable completeness: the set $\{A_jP_i\}$ spans the local operator space, ensuring that all coefficient directions are probed.
\item Band coverage: at least $2M{+}1$ distinct band indices $k$ are included.
\end{enumerate}
Under these operations, once $\mathbf A$ has full column rank, the reconstruction $\mathbf c$ is unique in the noiseless regime.

Then we turn to a more realistic setting, where measurement outcomes are subject to statistical noise and the Fourier expansion of $H(t)$ is truncated.  
Our goal is to determine, with high probability, a sufficient number of independent linear constraints (rows in the system $\mathbf{A}\mathbf{c}=\boldsymbol{\beta}$) to guarantee that the least-squares estimator $\widehat{\mathbf{c}}$ is close to the true coefficient vector $\mathbf{c}^\star$.

\paragraph*{Notation.}
Throughout, $\|\cdot\|$ denotes the operator/spectral norm for matrices and the Euclidean norm for vectors;
$\|\cdot\|_{\mathrm F}$ is the Frobenius norm.
For a symmetric matrix $M$, $\lambda_{\min}(M)$ and $\lambda_{\max}(M)$ denote its extreme eigenvalues; 
for a (possibly rectangular) matrix $A$, $\sigma_{\min}(A)$ denotes its smallest singular value.
For a real random variable $X$, $\|X\|_{\psi_2}$ and $\|X\|_{\psi_1}$ are its sub-Gaussian and sub-exponential Orlicz norms, respectively.
We write $\mathbb S^{K-1}=\{u\in\mathbb R^K:\|u\|=1\}$.

\begin{theorem}
\label{thm:T1-Bernstein}
Consider the linear system $\mathbf A\mathbf c=\boldsymbol\beta$ of Eq.~\eqref{eq:master_block} with unknown
$\mathbf c^\star\in\mathbb R^K$, where $K=(2M{+}1)R$.
Assume:
\begin{enumerate}[label=(A\arabic*)]
\item \textbf{Sampling.} The stacked rows of $\mathbf A$ use at least $2M{+}1$ distinct bands $k$.
\item \textbf{Design regularity.} The rows $a_r^\top\in\mathbb R^{K}$ are i.i.d., mean-zero, with covariance
$\boldsymbol\Sigma:=\mathbb E[a_r a_r^\top]$ satisfying $\mu_0\mathbf I\preceq\boldsymbol\Sigma\preceq \mu_1\mathbf I$ for some $0<\mu_0\le \mu_1<\infty$.
Moreover, for every unit $u\in\mathbb S^{K-1}$,
\[
\|\langle a_r,u\rangle\|_{\psi_2}\ \le\ \psi\,\sqrt{u^\top\boldsymbol\Sigma\,u}.
\]
\item \textbf{Noise and model mismatch.} $\boldsymbol\beta=\mathbf A\mathbf c^\star+\boldsymbol\eta+\boldsymbol\varepsilon$, 
where the noise $\eta_r$ are i.i.d.\ mean-zero sub-Gaussian with proxy $\sigma^2$ (e.g.\ $\sigma^2=\Theta(1/N_s)$ for $N_s$ shots),
$\boldsymbol\eta$ is independent of $\{a_r\}$, and $\boldsymbol\varepsilon$ collects neglected Fourier tails (Eq.~\eqref{eq:H_general}); set $\varepsilon_0:=\|\boldsymbol\varepsilon\|_\infty$.
\end{enumerate}
Let $\widehat{\mathbf c}$ be the least-squares estimator. There exist absolute constants $C,c>0$ such that, for any $\delta\in(0,1)$, 
if
\begin{equation}\label{eq:S-bern-main}
S\ \ge\ C\,\frac{\psi^4\mu_1^2}{\mu_0^2}\;K\;\bigl(\log(2K)+\log(1/\delta)\bigr),
\end{equation}
then with probability at least $1-2\delta$,
\begin{align}
\label{eq:bern-smin}
\sigma_{\min}(\mathbf A)\ &\ge\ \sqrt{S\,\mu_0/2},\\[2mm]
\label{eq:bern-error}
\|\widehat{\mathbf c}-\mathbf c^\star\|_2\ &\le\ C\,\frac{\sigma}{\mu_0}\sqrt{\frac{\mu_1\,\bigl(K+\log(1/\delta)\bigr)}{S}}
\;+\;2\,\frac{\sqrt{\mu_1}}{\mu_0}\,\varepsilon_0.
\end{align}
If $\|H_m\|=O(|m|^{-(p+1)})$ (resp.\ $O(e^{-\gamma|m|})$), then $\varepsilon_0=O(M^{-p})$ (resp.\ $O(e^{-\gamma M})$).
\end{theorem}

\begin{proof}
Define the sample covariance 
\(\widehat{\boldsymbol\Sigma}:=\frac{1}{S}\mathbf A^\top\mathbf A=\frac{1}{S}\sum_{r=1}^S a_r a_r^\top\)
and set $X_r:=a_r a_r^\top-\boldsymbol\Sigma$ so that 
\[\widehat{\boldsymbol\Sigma}-\boldsymbol\Sigma=\frac{1}{S}\sum_{r=1}^S X_r.\]

\paragraph*{Step 1 (Amplitude and variance parameters for matrix Bernstein).}
We use the matrix Bernstein inequality (Observation~\ref{lem:matrix-bernstein}) for the sum $\sum_{r=1}^S X_r$ of independent, centered, self-adjoint matrices.

\emph{Amplitude.} For any unit $u$,
\[
\xi_r(u):=u^\top X_r u=\langle a_r,u\rangle^2-\mathbb E\langle a_r,u\rangle^2
\]
is sub-exponential because the square of a sub-Gaussian is sub-exponential(Observation~\ref{lem:sg-square}). Using (A2),
\[
\|\xi_r(u)\|_{\psi_1}\ \le\ C\,\|\langle a_r,u\rangle\|_{\psi_2}^2
\ \le\ C\,\psi^2\,u^\top\boldsymbol\Sigma u
\ \le\ C\,\psi^2\,\mu_1.
\]
Hence a uniform amplitude parameter is
\begin{equation}\label{eq:Bdef}
B\ :=\ \sup_{\|u\|=1}\|\xi_r(u)\|_{\psi_1}\ \le\ C\,\psi^2\,\mu_1.
\end{equation}

\emph{Variance proxy.} Let $v:=\bigl\|\sum_{r=1}^S \mathbb E[X_r^2]\bigr\|$.
A direct calculation (Observation~\ref{lem:variance-proxy}) yields the matrix inequality
\begin{equation}\label{eq:Ex2_master_B}
\mathbb E[X_r^2]\ \preceq\ C\,\psi^4\,(\operatorname{tr}\boldsymbol\Sigma)\,\boldsymbol\Sigma,
\qquad\Rightarrow\qquad
v\ \le\ C\,S\,\psi^4\,(\operatorname{tr}\boldsymbol\Sigma)\,\|\boldsymbol\Sigma\|
\ \le\ C\,S\,\psi^4\,K\,\mu_1^2,
\end{equation}
using $\operatorname{tr}\boldsymbol\Sigma\le K\mu_1$.

\paragraph*{Step 2 (Covariance concentration and spectral bounds).}
Matrix Bernstein (sub-exponential form, Observation~\ref{lem:matrix-bernstein}) gives for all $t>0$,
\[
\Pr\!\left(\left\|\sum_{r=1}^S X_r\right\|\ge t\right)
\ \le\ 2K\;\exp\!\left[-\,c\,\min\!\left(\frac{t^2}{v},\ \frac{t}{B}\right)\right].
\]
Set $t=\frac{S\mu_0}{2}$ and substitute \eqref{eq:Bdef} and \eqref{eq:Ex2_master_B}. One gets
\[
\Pr\!\left(\bigl\|\widehat{\boldsymbol\Sigma}-\boldsymbol\Sigma\bigr\|>\frac{\mu_0}{2}\right)
\ \le\ 2K\;\exp\!\left\{-\,c\,S\cdot
\min\!\left(\frac{\mu_0^2}{C\,\psi^4\,\mu_1^2\,K},\ \frac{\mu_0}{C\,\psi^2\,\mu_1}\right)\right\}.
\]
Thus it suffices to choose $S$ so that both
\begin{align}
S\ &\ge\ C\,\frac{\psi^4\mu_1^2}{\mu_0^2}\;K\;\bigl(\log(2K)+\log(1/\delta)\bigr),\label{eq:S-variance-branch}\\
S\ &\ge\ C\,\frac{\psi^2\mu_1}{\mu_0}\;\bigl(\log(2K)+\log(1/\delta)\bigr)\label{eq:S-amplitude-branch}
\end{align}
hold; for moderate/large $K$, \eqref{eq:S-variance-branch} dominates, which is exactly \eqref{eq:S-bern-main}.

Let $\mathbf E:=\widehat{\boldsymbol\Sigma}-\boldsymbol\Sigma$. For any unit vector $u$,
\[
u^\top\widehat{\boldsymbol\Sigma}u
= u^\top\boldsymbol\Sigma u + u^\top\mathbf E u
\;\ge\; u^\top\boldsymbol\Sigma u - |u^\top\mathbf E u|
\;\ge\; \lambda_{\min}(\boldsymbol\Sigma) - \|E\|,
\]
Therefore, taking the minimum over all unit $u$ yields the bound
\begin{equation}\label{eq:weyl-min}
\lambda_{\min}(\widehat{\boldsymbol\Sigma})
\;\ge\; \lambda_{\min}(\boldsymbol\Sigma) - \|\widehat{\boldsymbol\Sigma}-\boldsymbol\Sigma\|.
\end{equation}
By Assumption $\lambda_{\min}(\boldsymbol\Sigma)\ge \mu_0$. and results from 
$\|\widehat{\boldsymbol\Sigma}-\boldsymbol\Sigma\|\le \mu_0/2$, we have
\[
\lambda_{\min}(\widehat{\boldsymbol\Sigma})
\;\ge\; \mu_0 - \frac{\mu_0}{2}
\;=\; \frac{\mu_0}{2}.
\]
An analogous argument from the other side gives
\begin{equation}\label{eq:weyl-max}
\lambda_{\max}(\widehat{\boldsymbol\Sigma})
\;\le\; \lambda_{\max}(\boldsymbol\Sigma) + \|\widehat{\boldsymbol\Sigma}-\boldsymbol\Sigma\|
\;\le\; \mu_1 + \frac{\mu_0}{2}
\;\le\; C\,\mu_1,
\end{equation}
where the last inequality absorbs the constant $1+\tfrac{\mu_0}{2\mu_1}\le C$ into $C$.
Finally, with probability $\ge 1-\delta$,
\begin{equation}\label{eq:sminA-conc}
\sigma_{\min}(\mathbf A)=\sqrt{\lambda_{\min}(\mathbf A^\top\mathbf A)}=\sqrt{S\,\lambda_{\min}(\widehat{\boldsymbol\Sigma})}
\ \ge\ \sqrt{S\,\mu_0/2},\qquad
\|\mathbf A\|=\sqrt{\lambda_{\max}(\mathbf A^\top\mathbf A)}\ \le\ \sqrt{C\,S\,\mu_1}.
\end{equation}
This proves \eqref{eq:bern-smin}.

\paragraph*{Step 3 (Noise term via an $\varepsilon$-net).}
Write the LS error as
\[
\widehat{\mathbf c}-\mathbf c^\star=(\mathbf A^\top\mathbf A)^{-1}\mathbf A^\top(\boldsymbol\eta+\boldsymbol\varepsilon).
\]
Set $\mathbf z:=\mathbf A^\top\boldsymbol\eta=\sum_{r=1}^S \eta_r a_r\in\mathbb R^K$ and condition on $\mathbf A$ (hence on $\{a_r\}$).
By independence in (A3), for any unit $u\in\mathbb S^{K-1}$,
\[
\langle u,\mathbf z\rangle=\sum_{r=1}^S \eta_r\,\langle a_r,u\rangle
\]
is sub-Gaussian with
\[
\|\langle u,\mathbf z\rangle\|_{\psi_2}\ \le\ C\,\sigma\,\Big(\sum_{r=1}^S \langle a_r,u\rangle^2\Big)^{1/2}
\ =\ C\,\sigma\,\|\mathbf A u\|
\ \le\ C\,\sigma\,\|\mathbf A\|.
\]
Hence, for all $t>0$,
\begin{equation}\label{eq:fixed-u-tail}
\Pr\big(|\langle u,\mathbf z\rangle|\ge t\ \big|\ \mathbf A\big)
\ \le\ 2\exp\!\Big(-c\,\frac{t^2}{\sigma^2\,\|\mathbf A\|^2}\Big).
\end{equation}
Let $\mathcal N\subset\mathbb S^{K-1}$ be a $1/2$-net with $|\mathcal N|\le 5^K$. By the net lifting lemma(specified in Observation~\ref{lem:net-union-Ateta}),
$\|\mathbf z\|\le 2\max_{u\in\mathcal N}|\langle u,\mathbf z\rangle|$.
A union bound applied to \eqref{eq:fixed-u-tail} over $u\in\mathcal N$ yields, upon choosing
\[
t\ :=\ C\,\sigma\,\|\mathbf A\|\,\sqrt{K+\log(1/\delta)}\ ,
\]
that
\begin{equation}\label{eq:Ateta-net}
\|\mathbf A^\top\boldsymbol\eta\|
\ =\ \|\mathbf z\|
\ \le\ C\,\sigma\,\|\mathbf A\|\,\sqrt{K+\log(1/\delta)}\ ,
\qquad \text{with probability at least } 1-\delta \text{ (conditional on $\mathbf A$)}.
\end{equation}
Deconditioning and intersecting with the Step~2 event \eqref{eq:sminA-conc}, we obtain
\[
\Big\|(\mathbf A^\top\mathbf A)^{-1}\mathbf A^\top\boldsymbol\eta\Big\|
\ \le\ \frac{\|\mathbf A^\top\boldsymbol\eta\|}{\sigma_{\min}(\mathbf A^\top\mathbf A)}
\ \le\ C\,\frac{\sigma}{\mu_0}\,\sqrt{\frac{\mu_1\,\bigl(K+\log(1/\delta)\bigr)}{S}},
\]
with probability at least $1-\delta$.

\paragraph*{Step 4 (Model mismatch).}
On the same event \eqref{eq:sminA-conc}, using $\|\mathbf A^\top\boldsymbol\varepsilon\|\le\|\mathbf A\|\,\|\boldsymbol\varepsilon\|_2$ and $\|\boldsymbol\varepsilon\|_2\le \sqrt{S}\,\varepsilon_0$,
\[
\Bigl\|(\mathbf A^\top\mathbf A)^{-1}\mathbf A^\top\boldsymbol\varepsilon\Bigr\|
\ \le\ \frac{\|\mathbf A\|\,\sqrt{S}\,\varepsilon_0}{S\mu_0/2}
\ \le\ 2\,\frac{\sqrt{\mu_1}}{\mu_0}\,\varepsilon_0.
\]

\paragraph*{Step 5 (Union bound).}
Let $E_{\rm des}$ be the Step~2 design event and $E_{\rm noise}$ the Step~3 noise event. Then
$\Pr(E_{\rm des})\ge 1-\delta$ and $\Pr(E_{\rm noise})\ge 1-\delta$, hence
$\Pr(E_{\rm des}\cap E_{\rm noise})\ge 1-2\delta$, which establishes \eqref{eq:bern-error}.

\end{proof}

\begin{corollary}
\label{cor:bern}
Choose $M$ such that $2\sqrt{\mu_1}\varepsilon_0/\mu_0\le \varepsilon/2$.
If
\[
S\ \ge\ C\,\frac{\psi^4\mu_1^2}{\mu_0^2}\;K\;\bigl(\log(2K)+\log(1/\delta)\bigr)
\quad\text{and}\quad
S\ \ge\ C\,\frac{\sigma^2\mu_1}{\mu_0^2}\,\frac{K+\log(1/\delta)}{(\varepsilon/2)^2},
\]
then $\|\widehat{\mathbf c}-\mathbf c^\star\|_2\le \varepsilon$ with probability at least $1-2\delta$.
\end{corollary}

\subsection{Supporting Observations}
\begin{observation}[Square of a sub-Gaussian is sub-exponential, Lemma 2.8.6 in~\cite{vershynin2018high}]\label{lem:sg-square}
If $X$ is centered sub-Gaussian with $\|X\|_{\psi_2}\le K$, then $X^2-\mathbb E[X^2]$ is sub-exponential and
$\|X^2-\mathbb E[X^2]\|_{\psi_1}\le C\,K^2$.
\end{observation}

\begin{observation}[Moment growth for sub-Gaussians, Proposation 2.6.1 in~\cite{vershynin2018high}]\label{lem:sg-moment}
If $X$ is sub-Gaussian with $\|X\|_{\psi_2}\le K$, then for all $q\ge 1$,
$(\mathbb E|X|^q)^{1/q}\le C\,K\,\sqrt{q}$.
\end{observation}

\begin{observation}[Hanson--Wright inequality~\cite{rudelson2013hanson}]\label{lem:HW}
Let $g=(g_1,\dots,g_n)$ have independent centered sub-Gaussian entries with $\|g_i\|_{\psi_2}\le \psi$ and let $A=A^\top$.
Then for all $t>0$,
\[
\Pr\!\big(|g^\top A g-\mathbb E[g^\top A g]|>t\big)\ \le\
2\exp\!\left[-c\,\min\!\left(\frac{t^2}{\psi^4\|A\|_{\mathrm F}^2},\ \frac{t}{\psi^2\|A\|}\right)\right].
\]
\end{observation}

\begin{observation}[Tail-to-moment identity]\label{lem:tail-moment}
If $X$ is centered with $\mathbb E[X^2]<\infty$, then
$\ \mathbb E[X^2]=\int_{0}^{\infty} 2t\,\Pr(|X|>t)\,dt$.
\end{observation}
\begin{proof}
    This is ignored as a standard calculation can be done with partial integration. 
\end{proof}

\begin{observation}[Covering number of the sphere, Corollary 4.2.11 in ~\cite{vershynin2018high}]\label{lem:sphere-cover}
For $\varepsilon\in(0,1)$ there exists an $\varepsilon$-net $\mathcal N_\varepsilon\subset\mathbb S^{K-1}$ with
$|\mathcal N_\varepsilon|\le (1+2/\varepsilon)^K$. In particular, for $\varepsilon=\tfrac12$ one has $|\mathcal N_{1/2}|\le 5^K$.
\emph{Reference:} Vershynin (2018), Lemma 5.2.
\end{observation}

\begin{observation}[Net lifting, Lemma 4.4.1 in \cite{vershynin2018high}]\label{lem:epsilon-net}
Let $z\in\mathbb R^K$ and let $\mathcal N_\varepsilon\subset \mathbb S^{K-1}$ be an $\varepsilon$-net. Then
\[
\|z\|\ \le\ \frac{1}{1-\varepsilon}\ \max_{u\in\mathcal N_\varepsilon}|\langle u,z\rangle|.
\]
In particular, for $\varepsilon=\tfrac12$,
$\ \|z\|\le 2\max_{u\in\mathcal N_{1/2}}|\langle u,z\rangle|$.
\emph{Reference:} Vershynin (2018), Lemma 5.4.
\end{observation}


\begin{observation}[Matrix Bernstein, sub-exponential version, Theorem 6.2 in~\cite{TroppUserFriendly}]\label{lem:matrix-bernstein}
Let $\{X_r\}_{r=1}^S$ be independent, centered, self-adjoint random matrices of the same size and assume
$\|u^\top X_r u\|_{\psi_1}\le B$ for all unit $u$, and let
$v:=\big\|\sum_{r=1}^S \mathbb E[X_r^2]\big\|$.
Then for all $t>0$,
\[
\Pr\!\left(\left\|\sum_{r=1}^S X_r\right\|\ge t\right)\ \le\ 2d\ \exp\!\left[-c\,\min\!\left(\frac{t^2}{v},\ \frac{t}{B}\right)\right],
\]
where $d$ is the matrix dimension.
\end{observation}


\begin{observation}[Variance proxy bound for the covariance sum]\label{lem:variance-proxy}
Let $a\in\mathbb R^K$ be a centered \emph{relative sub-Gaussian} vector with covariance $\Sigma=\mathbb E[aa^\top]$
satisfying $\mu_0 I\preceq \Sigma\preceq \mu_1 I$ and
$\|\langle a,u\rangle\|_{\psi_2}\le \psi\,\sqrt{u^\top\Sigma u}$ for all unit $u$.
Set $X:=aa^\top-\Sigma$. Then
\[
\mathbb E[X^2]\ \preceq\ C\,\psi^4\,(\operatorname{tr}\Sigma)\,\Sigma.
\]
Consequently, for independent copies $\{a_r\}$ and $X_r:=a_ra_r^\top-\Sigma$,
\[
\left\|\sum_{r=1}^S \mathbb E[X_r^2]\right\|\ \le\ C\,S\,\psi^4\,(\operatorname{tr}\Sigma)\,\|\Sigma\|
\ \le\ C\,S\,\psi^4\,K\,\mu_1^2.
\]
\end{observation}

\begin{proof}
Fix a unit $u$. Using Cauchy–Schwarz,
\[
u^\top\mathbb E[X^2]u=\mathbb E\!\big[(a^\top u)^2\,\|a\|^2\big]-(u^\top\Sigma u)^2
\ \le\ \bigl(\mathbb E(a^\top u)^4\bigr)^{1/2}\,\bigl(\mathbb E\|a\|^4\bigr)^{1/2}.
\]
By Observation~\ref{lem:sg-moment}, we apply to $X=a^\top u$, thus
\[\|X\|_{\psi_2}\le \psi\,\sqrt{u^\top\Sigma u}, \quad (\mathbb E(a^\top u)^4)^{1/2}\le C\,\psi^2\,(u^\top\Sigma u).\]
Let $g:=\Sigma^{-1/2}a$ (isotropic sub-Gaussian, $\|g_i\|_{\psi_2}\lesssim \psi$). Then
$\|a\|^2=g^\top\Sigma g$ and $\mathbb E\|a\|^2=\operatorname{tr}\Sigma$. By Hanson–Wright (Observation~\ref{lem:HW}) and the tail-to-moment identity (Observation~\ref{lem:tail-moment}),
$\operatorname{Var}(\|a\|^2)\le C\,\psi^4\,\|\Sigma\|_{\mathrm F}^2$.
Hence
\[
\bigl(\mathbb E\|a\|^4\bigr)^{1/2}
=\bigl((\operatorname{tr}\Sigma)^2+\operatorname{Var}(\|a\|^2)\bigr)^{1/2}
\ \le\ \operatorname{tr}\Sigma + C^{1/2}\psi^2\|\Sigma\|_{\mathrm F}
\ \le\ C'\,\psi^2\,\operatorname{tr}\Sigma,
\]
since $\|\Sigma\|_{\mathrm F}\le \operatorname{tr}\Sigma$ for $\Sigma\succeq 0$.

Combining the two factors gives
\[u^\top\mathbb E[X^2]u\le C\,\psi^4\,(\operatorname{tr}\Sigma)\,(u^\top\Sigma u)\] 
for all unit $u$, i.e.
$\mathbb E[X^2]\preceq C\,\psi^4\,(\operatorname{tr}\Sigma)\,\Sigma$, which implies \eqref{eq:Ex2_master_B}. The norm bound then follows from
$\|\sum_r \mathbb E[X_r^2]\|\le \sum_r \|\mathbb E[X_r^2]\|$ and $\|\Sigma\|=\mu_1$, $\operatorname{tr}\Sigma\le K\mu_1$.
\end{proof}

\begin{observation}[Net lifting, union bound for $\|\mathbf A^\top\boldsymbol\eta\|$]\label{lem:net-union-Ateta}
Let $\mathbf A\in\mathbb R^{S\times K}$ be fixed and let $\boldsymbol\eta=(\eta_1,\dots,\eta_S)$ have independent, centered, sub-Gaussian entries with
$\|\eta_r\|_{\psi_2}\le \sigma$. Set $\mathbf z:=\mathbf A^\top\boldsymbol\eta$. Then, for any $\delta\in(0,1)$,
\[
\Pr\!\Big(\|\mathbf z\|\ \le\ C\,\sigma\,\|\mathbf A\|\,\sqrt{K+\log(1/\delta)}\ \Big|\ \mathbf A\Big)\ \ge\ 1-\delta.
\]
\end{observation}

\begin{proof}

Recall we condition on $\mathbf A$ and set $\mathbf z:=\mathbf A^\top\boldsymbol\eta\in\mathbb R^K$.
For any unit $u\in\mathbb S^{K-1}$, we already showed the conditional tail
\begin{equation}\label{eq:fixed-u-tail-again}
\Pr\big(|\langle u,\mathbf z\rangle|\ge t\ \big|\ \mathbf A\big)
\ \le\ 2\exp\!\Big(-c\,t^2/(\sigma^2\|\mathbf A\|^2)\Big),\qquad t>0.
\end{equation}

\paragraph*{Step A: $\varepsilon$-net on the sphere, Observation \ref{lem:sphere-cover} ---} 
Let $\mathcal N\subset\mathbb S^{K-1}$ be a $1/2$-net, i.e., for every $u\in\mathbb S^{K-1}$ there exists $v\in\mathcal N$
with $\|u-v\|_2\le 1/2$. A standard volumetric argument gives the cardinality bound
\begin{equation}\label{eq:net-card}
|\mathcal N|\ \le\ (1+2/\varepsilon)^K\ \Big|_{\ \varepsilon=1/2}\ \le\ 5^K.
\end{equation}

\paragraph*{Step B: Net lifting, Observation~\ref{lem:epsilon-net}---}
For any $x\in\mathbb R^K$ and any $1/2$-net $\mathcal N\subset\mathbb S^{K-1}$,
\begin{equation}\label{eq:lifting}
\|x\|\ \le\ \frac{1}{1-\varepsilon}\,\max_{v\in\mathcal N} |\langle v,x\rangle|\ \Bigg|_{\ \varepsilon=1/2}\ \le\ 2\,\max_{v\in\mathcal N} |\langle v,x\rangle|.
\end{equation}
Applying \eqref{eq:lifting} to $x=\mathbf z$ yields
\begin{equation}\label{eq:z-net}
\|\mathbf z\|\ \le\ 2\,\max_{v\in\mathcal N} |\langle v,\mathbf z\rangle|.
\end{equation}

\paragraph*{Step C: Union bound over the net---}
Using \eqref{eq:fixed-u-tail-again} for each fixed $v\in\mathcal N$ and the union bound,
\[
\Pr\!\Big(\|\mathbf z\|\ge 2t\ \Big|\ \mathbf A\Big)
\ \le\ \sum_{v\in\mathcal N}\Pr\!\big(|\langle v,\mathbf z\rangle|\ge t\ \big|\ \mathbf A\big)
\ \le\ |\mathcal N|\cdot 2\exp\!\Big(-c\,t^2/(\sigma^2\|\mathbf A\|^2)\Big).
\]
Using \eqref{eq:net-card}, we get
\begin{equation}\label{eq:tail-over-net}
\Pr\!\Big(\|\mathbf z\|\ge 2t\ \Big|\ \mathbf A\Big)
\ \le\ 2\cdot 5^K\ \exp\!\Big(-c\,t^2/(\sigma^2\|\mathbf A\|^2)\Big).
\end{equation}

\paragraph*{Step D: Choice of the threshold $t$---}
We choose $t$ so that the RHS of \eqref{eq:tail-over-net} is at most $\delta$, i.e.
\[
2\cdot 5^K\ \exp\!\Big(-c\,t^2/(\sigma^2\|\mathbf A\|^2)\Big)\ \le\ \delta
\quad\Longleftrightarrow\quad
\frac{t^2}{\sigma^2\|\mathbf A\|^2}\ \ge\ \frac{K\log 5+\log 2+\log(1/\delta)}{c}.
\]
Thus any
\[
t\ \ge\ \sigma\,\|\mathbf A\|\ \sqrt{\frac{K\log 5+\log 2+\log(1/\delta)}{c}}
\]
will do. Absorbing the absolute constants into $C$ (and noting $K\log 5+\log 2\le C(K+1)$) yields the convenient form
\begin{equation}\label{eq:t-choice}
t\ :=\ C\,\sigma\,\|\mathbf A\|\,\sqrt{K+\log(1/\delta)}.
\end{equation}
Combining \eqref{eq:z-net}–\eqref{eq:t-choice} we obtain, conditional on $\mathbf A$,
\begin{equation}\label{eq:Ateta-net-final}
\|\mathbf A^\top\boldsymbol\eta\|=\|\mathbf z\|\ \le\ C\,\sigma\,\|\mathbf A\|\,\sqrt{K+\log(1/\delta)}\qquad\text{with probability at least }1-\delta.
\end{equation}

\end{proof}

\end{document}